\documentclass[draft]{birkjour}
\usepackage[noadjust]{cite}
\usepackage{xcolor}
\usepackage{amssymb}
\RequirePackage[all]{xy}

\usepackage{longtable}

\newtheorem{thm}{Theorem}[section]

\theoremstyle{definition}

\theoremstyle{remark}

\newtheorem*{ex}{Example}
\numberwithin{equation}{section}

\newcommand{\BibTeX}{B\kern-0.1emi\kern-0.017emb\kern-0.15em\TeX}
\newcommand{\XYpic}{$\mathrm{X\kern-0.3em\raisebox{-0.18em}{Y}}$-$\mathrm{pic}\,$}

\newcommand{\cl}{C \kern -0.1em \ell}  

\newcommand{\BR}{\mathbb{R}}
\newcommand{\BC}{\mathbb{C}}
\newcommand{\BH}{\mathbb{H}}

\newcommand{\ed}{\end{document}}

\newcommand{\Mat}{{\rm Mat}}

\newcommand{\spn}{{\rm span}}
\newcommand{\diag}{{\rm diag}}

\def\cl{\mathcal {G}}
\newcommand{\U}{{\rm U}}
\newcommand{\G}{{\rm G}}
\newcommand{\T}{{\rm T}}
\newcommand{\HH}{{\rm H}}
\newcommand{\OO}{{\rm O}}
\newcommand{\Sp}{{\rm Sp}}

\begin{document}

%
%
%
%
%
%
%
%
%

\title[On SVD and Polar Decomposition in Clifford Algebras]
 {On SVD and Polar Decomposition in Real and Complexified Clifford Algebras}

\author[D. Shirokov]{Dmitry Shirokov}
\address{%
HSE University\\
Moscow 101000\\
Russia
\medskip}
\address{
and
\medskip}
\address{
Institute for Information Transmission Problems of the Russian Academy of Sciences \\
Moscow 127051 \\
Russia}
\email{dm.shirokov@gmail.com}

%
\begin{abstract}
In this paper, we present a natural implementation of singular value decomposition (SVD) and polar decomposition of an arbitrary multivector in nondegenerate real and complexified Clifford geometric algebras of arbitrary dimension and signature. The new theorems involve only operations in geometric algebras and do not involve matrix operations. We naturally define these and other related structures such as Hermitian conjugation, Euclidean space, and Lie groups in geometric algebras. The results can be used in various applications of geometric algebras in computer science, engineering, and physics.
\end{abstract}
\label{page:firstblob}
\maketitle

\section{Introduction}\label{sectIntro}

 The method of singular value decomposition for matrices was discovered independently by E. Beltrami in 1873 \cite{Beltrami} and C. Jordan in 1874 \cite{Jordan1,Jordan2}. This method is classical and is widely used in various applications -- signal and image processing, least squares fitting of data, process control, computer science, engineering, big data, machine learning, physics, etc. A lot of literature is devoted to SVD, its algorithms, and its applications (see, for example, the books by  G. Golub \& C. Van Loan \cite{Van} and G. Forsythe,  M. Malcolm \&  C. Moler \cite{For}, the paper \cite{Golubb}). Another generalization of SVD, quaternion SVD, was invented by F. Zhang in his paper \cite{quat}; for quaternion SVD algorithms, see \cite{quat2,Sangwine}. Applications of quaternion SVD in image processing are considered in \cite{image}.  Another generalization of SVD, hyperbolic SVD, and its applications are discussed in \cite{Boj, Zha, noteHSVD, YM, HSVD}.
 Polar decomposition of complexified quaternions and octonions is discussed in \cite{polar1}. See the survey \cite{Signal} on different applications of SVD and its generalizations (in particular, quaternion SVD and octonion SVD) in signal and image processing. Note also the paper \cite{Abl} on computing SVD of multivectors using the corresponding matrix representations (an explicit implementation of SVD in Clifford algebras is not presented, in contrast to our work).

In the present paper, we present a natural implementation of singular value decomposition (SVD) and polar decomposition of an arbitrary multivector in nondegenerate real and complexified Clifford geometric algebras of arbitrary dimension and signature. We naturally define these and other related structures such as Hermitian conjugation, Euclidean space, and Lie groups in geometric algebras. ``Natural'' means that our definitions and statements involve only operations in geometric algebra and do not involve the corresponding matrix representations. The meaning of SVD in Clifford geometric algebras is the following: after multiplication on the left and on the right by elements of a fixed Lie group, any multivector can be placed in a real subspace of lower dimension. Note that we present existing theorems and do not discuss algorithms in this paper. Theorems \ref{th1}, \ref{th2}, \ref{th1C}, and \ref{th2C} are new.

The paper is organized as follows. In Section \ref{sectGA}, we recall basic facts on nondegenerate real Clifford geometric algebras $\cl_{p,q}$. In Section \ref{sectES}, we introduce Hermitian conjugation and Euclidean space in $\cl_{p,q}$. In Section \ref{sectMR}, we present explicit matrix representation of $\cl_{p,q}$ and consider several Lie groups, which are isomorphic to orthogonal, unitary, and symplectic classical matrix Lie groups. In Section \ref{sectSVDM}, we recall classical theorems on SVD of real, complex, and quaternion matrices. In Section \ref{sectSVD}, we present SVD of multivectors in $\cl_{p,q}$ and several examples. In Section \ref{sectCPD}, we recall classical theorems on polar decomposition of real, complex, and quaternion matrices. In Section \ref{sectPD}, we present polar decomposition of multivectors in $\cl_{p,q}$. In Section \ref{sectCompl}, we consider the case of complexified Clifford geometric algebras $\BC\otimes\cl_{p,q}$ and present SVD and polar decomposition of multivectors in $\BC\otimes\cl_{p,q}$. Conclusions follow in Section \ref{sectConcl}.

This paper is an extended version of the short note (11 pages) in Conference Proceedings \cite{SVDENGAGE} (Empowering Novel Geometric Algebra for Graphics \& Engineering Workshop within the International Conference Computer Graphics International 2023). Sections \ref{sectIntro}, \ref{sectMR}, \ref{sectSVD}, and \ref{sectConcl} are extended; Section \ref{sectCompl} and Theorems \ref{th1C} and \ref{th2C} are new.

\section{Real Clifford geometric algebra}\label{sectGA}

Let us consider the real Clifford geometric algebra (GA) $\cl_{p,q}$ \cite{Hestenes,Lounesto,Doran,Bulg} with the identity element $e\equiv 1$ and the generators $e_a$, $a=1, 2, \ldots, n$, where $n=p+q\geq 1$. The generators satisfy the conditions
$$
e_a e_b+e_b e_a=2\eta_{ab}e,\qquad \eta=(\eta_{ab})=\diag(\underbrace{1, \ldots , 1}_p, \underbrace{-1, \ldots, -1}_{q}).
$$
Consider the subspaces $\cl^k_{p,q}$ of grades $k=0, 1, \ldots, n$, which elements are linear combinations of the basis elements $e_A=e_{a_1 a_2 \ldots a_k}=e_{a_1}e_{a_2}\cdots e_{a_k}$, $1 \leq a_1<a_2<\cdots< a_k \leq n$, with ordered multi-indices of length $k$. An arbitrary element (multivector) $M\in\cl_{p,q}$ has the form
$$
M=\sum_A m_A e_A\in\cl_{p,q},\qquad m_A\in\BR,
$$
where we have a sum over arbitrary multi-index $A$ of length from $0$ to $n$.
The projection of $M$ onto the subspace $\cl^k_{p,q}$ is denoted by $\langle M \rangle_k$.

The grade involution and reversion of a multivector $M\in\cl_{p,q}$ are denoted by 
\begin{eqnarray}
\widehat{M}=\sum_{k=0}^n(-1)^{k}\langle M \rangle_k,\qquad 
\widetilde{M}=\sum_{k=0}^n (-1)^{\frac{k(k-1)}{2}} \langle M \rangle_k\label{gir}
\end{eqnarray}
and have the properties
\begin{eqnarray}
\widehat{M_1 M_2}=\widehat{M_1} \widehat{M_2},\qquad \widetilde{M_1 M_2}=\widetilde{M_2} \widetilde{M_1},\qquad \forall M_1, M_2\in\cl_{p,q}.\label{invol}
\end{eqnarray}

\section{Euclidean space on GA}\label{sectES}

Let us consider an operation of Hermitian conjugation $\dagger$ in $\cl_{p,q}$  (see \cite{unitary,Bulg}):
\begin{eqnarray}
M^\dagger:=M|_{e_A \to (e_A)^{-1}}=\sum_A m_A (e_A)^{-1},\qquad M\in\cl_{p,q}.\label{dagger}
\end{eqnarray}
We have the following two other equivalent definitions of this operation:
\begin{eqnarray}
&&M^\dagger=\begin{cases}
e_{1\ldots p} \widetilde{M}e_{1\ldots p}^{-1}, & \mbox{if $p$ is odd,}\\
e_{1\ldots p}\widetilde{\widehat{M}}e_{1\ldots p}^{-1}, & \mbox{if $p$ is even,}\\
\end{cases}\\
&&M^\dagger=
\begin{cases}
e_{p+1\ldots n} \widetilde{M}e_{p+1\ldots n}^{-1}, & \mbox{if $q$ is even,}\\
e_{p+1\ldots n} \widetilde{\widehat{M}}e_{p+1\ldots n}^{-1}, & \mbox{if $q$ is odd.}\\
\end{cases}
\end{eqnarray}
The operation\footnote{Compare with the well-known operation $M_1 * M_2:=\langle \widetilde{M_1} M_2 \rangle_0$ in the real geometric algebra $\cl_{p,q}$, which is positive definite only in the case of signature $(p,q)=(n,0)$.} 
$$(M_1, M_2):=\langle M_1^\dagger M_2 \rangle_0,\qquad M_1, M_2\in\cl_{p,q}$$
is a (positive definite) scalar product with the properties
\begin{eqnarray}
&&\!\!\!\!\!\!\!\!\!\!\!\!\!\!\!(M_1, M_2)=(M_2, M_1),\\
&&\!\!\!\!\!\!\!\!\!\!\!\!\!\!\!(M_1+M_2, M_3)=(M_1, M_3)+(M_2, M_3),\quad (M_1, \lambda M_2)=\lambda (M_1, M_2),\\
&&\!\!\!\!\!\!\!\!\!\!\!\!\!\!\!(M, M)\geq 0,\quad \forall M\in\cl_{p,q};\qquad (M, M)=0 \Leftrightarrow M=0\label{||M||}
\end{eqnarray}
for arbitrary multivectors $M_1, M_2, M_3\in\cl_{p,q}$ and $\lambda\in\BR$.

Using this scalar product we introduce inner product space over the field of real numbers (Euclidean space) in $\cl_{p,q}$.

We have a norm 
\begin{eqnarray}
||M||:=\sqrt{(M,M)}=\sqrt{\langle M^\dagger M \rangle_0},\qquad M\in\cl_{p,q}\label{norm}
\end{eqnarray}
with the properties
\begin{eqnarray}
&&||M||\geq 0,\quad \forall M\in\cl_{p,q};\qquad ||M||=0 \Leftrightarrow M=0,\\
&&||M_1+M_2|| \leq ||M_1||+||M_2||,\qquad \forall M_1, M_2\in\cl_{p,q},\\
&&||\lambda M||=|\lambda| ||M||,\qquad \forall M\in\cl_{p,q},\qquad \forall \lambda\in\BR.
\end{eqnarray}

\section{Matrix representation of $\cl_{p,q}$ and Lie groups}\label{sectMR}

Let us consider the following faithful representation (isomorphism) of the real geometric algebra $\cl_{p,q}$
\begin{eqnarray}
\beta:\cl_{p,q}\,\,\to\,\, \begin{cases}
\Mat(2^{\frac{n}{2}},\BR), & \mbox{if $p-q=0, 2 \mod 8$,}\\
\Mat(2^{\frac{n-1}{2}},\BR)\oplus\Mat(2^{\frac{n-1}{2}},\BR), & \mbox{if $p-q=1 \mod 8$,}\\
\Mat(2^{\frac{n-1}{2}},\BC), & \mbox{if $p-q=3, 7 \mod 8$,}\\
\Mat(2^{\frac{n-2}{2}},\BH), & \mbox{if $p-q=4, 6 \mod 8$,}\\
\Mat(2^{\frac{n-3}{2}},\BH)\oplus\Mat(2^{\frac{n-3}{2}},\BH), & \mbox{if $p-q=5 \mod 8$.}
\end{cases}
\label{isom}
\end{eqnarray}
These isomorphisms are known as Cartan--Bott 8-periodicity.

Let us denote the size of the corresponding matrices by
\begin{eqnarray}
d:=\begin{cases}
2^{\frac{n}{2}}, & \mbox{if $p-q=0, 2 \mod 8$,}\\
2^{\frac{n+1}{2}}, & \mbox{if $p-q=1 \mod 8$,}\\
2^{\frac{n-1}{2}}, & \mbox{if $p-q=3, 5, 7 \mod 8$,}\\
2^{\frac{n-2}{2}}, & \mbox{if $p-q=4, 6 \mod 8$.}
\end{cases}\label{dd}
\end{eqnarray}
Note that we use block-diagonal matrices in the cases $p-q=1, 5\mod 8$. 

Let us present an explicit form of one of these representations of $\cl_{p,q}$ (see also \cite{Lie1,Lie2,LMA,det}). We denote this fixed representation by $\beta'$. For the identity element, we always use the identity matrix $\beta'(e)=I$ of the corresponding size $d$.  We always take $\beta'(e_{a_1 a_2 \ldots a_k}) = \beta' (e_{a_1}) \beta' (e_{a_2}) \cdots \beta'(e_{a_k})$. 

In some particular cases, we construct $\beta'$ in the following way:
\begin{itemize}
\item In the case $\cl_{0,1}$: $e_1 \to i$.
\item In the case $\cl_{1, 0}$: $e_1 \to \diag(1, -1)$.
\item In the case $\cl_{0, 2}$: $e_1 \to i$, $e_2 \to j$.
\item In the case $\cl_{0, 3}$: $e_1 \to \diag(i, -i)$, $e_2 \to \diag(j, -j)$, $e_3 \to \diag(k, -k)$.
\end{itemize}
Suppose we know $\beta'_a:=\beta'(e_a)$, $a = 1, \ldots, n$ for some fixed $\cl_{p,q}$, $p+q=n$. Then we construct explicit matrix representation of $\cl_{p+1, q+1}$, $\cl_{q+1, p-1}$, $\cl_{p-4, q-4}$ in the following way using the matrices $\beta'_a$, $a=1, \ldots, n$.
\begin{itemize}
\item In the case $\cl_{p+1, q+1}$: $e_a \to \diag(\beta'_a, -\beta'_a)$, $a=1, \ldots, p, p+2, \ldots, p+q+1$. In the subcase $p-q\neq 1 \mod 4$, we have 
$$
e_{p+1}\to \begin{pmatrix} 0 & I \\ I & 0 \end{pmatrix},\qquad e_{p+q+2} \to \begin{pmatrix} 0 & -I \\ I & 0  \end{pmatrix}.$$
 In the subcase $p-q=1\mod 4$, we have
$$
e_{p+1}\to \diag (\beta_1 \cdots \beta_n \Omega, -\beta_1 \cdots \beta_n \Omega), \qquad e_{p+q+2} \to \diag(\Omega, -\Omega),$$
where
\begin{eqnarray}
\Omega=\begin{pmatrix} 0 & -I \\ I & 0 \end{pmatrix}.\label{omega}
\end{eqnarray}
\item In the case $\cl_{q+1, p-1}$: $e_1 \to \beta'_1$, $e_i \to \beta'_i \beta'_1$, $i=2, \ldots, n$.
\item In the case $\cl_{p-4,q+4}$: $e_i \to \beta'_i \beta'_1 \beta'_2 \beta'_3 \beta'_4$, $i=1, 2, 3, 4$, $e_j \to \beta'_j$, $j=5, \ldots, n$.
\end{itemize}
Using these recurrences and the Cartan--Bott 8-periodicity, we obtain explicit matrix representation $\beta'$ of all $\cl_{p,q}$.

It can be directly verified that for this matrix representation we have
\begin{eqnarray}
\eta_{aa} \beta'(e_a)=\begin{cases}
(\beta'(e_a))^\T, & \mbox{if $p-q=0, 1, 2\mod 8$,}\\
(\beta'(e_a))^\HH, & \mbox{if $p-q=3, 7\mod 8$,}\\
(\beta'(e_a))^*, & \mbox{if $p-q=4, 5, 6\mod 8$,}\\
\end{cases}\qquad a=1, \ldots, n,
\end{eqnarray}
where $\T$ is transpose of a (real) matrix, $\HH$ is the Hermitian transpose of a (complex) matrix, $*$ is the conjugate transpose of a matrix over quaternions. Using the linearity, we get that these matrix conjugations are consistent with Hermitian conjugation of the corresponding multivector:
\begin{eqnarray}
\beta'(M^\dagger)=\begin{cases}
(\beta'(M))^\T, & \mbox{if $p-q=0, 1, 2\mod 8$,}\\
(\beta'(M))^\HH, & \mbox{if $p-q=3, 7\mod 8$,}\\
(\beta'(M))^*, & \mbox{if $p-q=4, 5, 6\mod 8$,}\\
\end{cases} \qquad M\in\cl_{p,q}.\label{sogl}
\end{eqnarray}

Note that the formulas like (\ref{sogl}) are not valid for an arbitrary matrix representation $\beta$ of the form (\ref{isom}). They are true for the matrix representations $\gamma=T^{-1}\beta' T$ obtained from $\beta'$ by the matrix $T$ such that
\begin{itemize}
\item $T^\T T= I$ in the cases $p-q=0, 1, 2\mod 8$,
\item $T^\HH T=I$ in the cases $p-q=3, 7 \mod 8$,
\item $T^* T=I$ in the cases $p-q=4, 5, 6\mod 8$.
\end{itemize}

Let us consider the following Lie group in $\cl_{p,q}$
\begin{eqnarray}
\G\cl_{p,q}=\{M\in \cl_{p,q}: M^\dagger M=e\},\label{GG}
\end{eqnarray}
where $\dagger$ is (\ref{dagger}). Note that all the basis elements $e_A$ of $\cl_{p,q}$ belong to this group by the definition.

Using (\ref{isom}) and (\ref{sogl}), we get the following isomorphisms of this group to the classical matrix Lie groups:
\begin{eqnarray}
\G\cl_{p,q}\simeq\begin{cases}
    \OO(2^{\frac{n}{2}}), &\mbox{if $p-q=0, 2 \mod 8$,}\\
    \OO(2^{\frac{n-1}{2}})\times\OO(2^{\frac{n-1}{2}}), &\mbox{if $p-q=1\mod 8$,}\\
    \U(2^{\frac{n-1}{2}}), &\mbox{if $p-q=3, 7 \mod 8$,}\\
    \Sp(2^{\frac{n-2}{2}}), &\mbox{if $p-q=4, 6 \mod 8$,}\\
    \Sp(2^{\frac{n-3}{2}})\times\Sp(2^{\frac{n-3}{2}}), &\mbox{if $p-q=5\mod 8$,}
\end{cases}\label{isgr}
\end{eqnarray}
where we have the following notation for (orthogonal, unitary, and simplectic correspondingly) classical matrix Lie groups
\begin{eqnarray}
\OO(k)=\{A\in\Mat(k, \BR):\quad A^\T A=I\},\\
\U(k)=\{A\in\Mat(k, \BC):\quad A^\HH A=I\},\\
\Sp(k)=\{A\in\Mat(k, \BH):\quad A^* A=I\}.
\end{eqnarray}
The group $\Sp(k)$ sometimes is called quaternionic unitary group or hyperunitary group. Note that this group also has the following realization in terms of complex matrices:
$$
\Sp(k)\simeq\{A\in\Mat (2k, \BC):\quad A^\T\Omega A= \Omega,\quad A^\HH A=I\},
$$
where $\Omega$ is (\ref{omega}).

The Lie algebra of the Lie group $\G\cl_{p,q}$ is 
$$\mathfrak{g}\cl_{p,q}=\{M\in\cl_{p,q}: M^\dagger=-M\}.$$
The basis of the Lie algebra $\mathfrak{g}\cl_{p,q}$ consists of anti-Hermitian basis elements
$$e_{p+1},\, \ldots,\, e_{n},\, e_{12},\, \ldots,\, e_{p-1 p},\, \ldots $$
The number of such elements can be calculated (see Theorem 7 in \cite{Bulg}); we get the dimension of the Lie group $\G\cl_{p,q}$ and the Lie algebra $\mathfrak{g}\cl_{p,q}$ (also we can calculate the dimension of the matrix Lie groups (\ref{isgr})):
\begin{eqnarray}
\dim(\G\cl_{p,q})=\dim(\mathfrak{g}\cl_{p,q})=2^{n-1}-2^{\frac{n-1}{2}}\sin\frac{\pi(p-q+1)}{4}\\
=\begin{cases}
2^{n-1}-2^{\frac{n-2}{2}}, & \mbox{if $p-q=0, 2 \mod 8$,}\\
2^{n-1}-2^{\frac{n-1}{2}}, & \mbox{if $p-q=1 \mod 8$,}\\
2^{n-1}, & \mbox{if $p-q=3, 7 \mod 8$,}\\
2^{n-1}+2^{\frac{n-2}{2}}, & \mbox{if $p-q=4, 6 \mod 8$,}\\
2^{n-1}+2^{\frac{n-1}{2}}, & \mbox{if $p-q=5 \mod 8$.}
\end{cases}\label{GGG}
\end{eqnarray}

\section{On the classical SVD of real, complex, and quaternion matrices}\label{sectSVDM}

We have the following well-known theorems on singular value decomposition of an arbitrary real, complex, and quaternion matrices (see, for example, \cite{For,Van,quat}). 

\begin{thm} For an arbitrary $A\in\BR^{n\times m}$, there exist matrices $U\in \OO(n)$ and $V\in\OO(m)$ such that
\begin{eqnarray}
 A=U\Sigma V^\T,\label{SVDR}
\end{eqnarray}
where
$$
\Sigma=\diag(\lambda_1, \lambda_2, \ldots, \lambda_k),\qquad k=\min(n, m),\qquad \BR\ni\lambda_1, \lambda_2, \ldots, \lambda_k\geq 0.
$$
Note that choosing matrices $U\in \OO(n)$ and $V\in\OO(m)$, we can always arrange diagonal elements of the matrix $\Sigma$
in decreasing order $\lambda_1\geq \lambda_2 \geq \cdots \geq \lambda_k\geq 0$.
\end{thm}

Diagonal elements of the matrix $\Sigma$ are called singular values, they are square roots of eigenvalues of the matrices $A A^\T$ or $A^\T A$. Columns of the matrices $U$ and $V$ are eigenvectors of the matrices $A A^\T$ and $A^\T A$ respectively.

\begin{thm} For an arbitrary $A\in\BC^{n\times m}$, there exist matrices $U\in \U(n)$ and $V\in\U(m)$ such that
\begin{eqnarray}
 A=U\Sigma V^\HH,\label{SVDC}
\end{eqnarray}
where
$$
\Sigma=\diag(\lambda_1, \lambda_2, \ldots, \lambda_k),\qquad k=\min(n, m),\qquad \BR\ni\lambda_1, \lambda_2, \ldots, \lambda_k\geq 0.
$$
Note that choosing matrices $U\in \U(n)$ and $V\in\U(m)$, we can always arrange diagonal elements of the matrix $\Sigma$
in decreasing order $\lambda_1\geq \lambda_2 \geq \cdots \geq \lambda_k\geq 0$.
\end{thm}

Diagonal elements of the matrix $\Sigma$ are called singular values, they are square roots of eigenvalues of the matrices $A A^\HH$ or $A^\HH A$. Columns of the matrices $U$ and $V$ are eigenvectors of the matrices $A A^\HH$ and $A^\HH A$ respectively.

\begin{thm} For an arbitrary $A\in\BH^{n\times m}$, there exist matrices $U\in \Sp(n)$ and $V\in\Sp(m)$ such that
\begin{eqnarray}
 A=U\Sigma V^*,\label{SVDH}
\end{eqnarray}
where
$$
\Sigma=\diag(\lambda_1, \lambda_2, \ldots, \lambda_k),\qquad k=\min(n, m),\qquad \BR\ni\lambda_1, \lambda_2, \ldots, \lambda_k\geq 0.
$$
\end{thm}

Diagonal elements of the matrix $\Sigma$ are called singular values.

\section{SVD in GA}\label{sectSVD}

In the following theorem, we present singular value decomposition of an arbitrary multivector in geometric algebra $\cl_{p,q}$. Note that the statement involves only operations in $\cl_{p,q}$.
 
\begin{thm}[SVD in GA]\label{th1} For an arbitrary multivector $M\in\cl_{p,q}$, there exist multivectors $U, V\in \G\cl_{p,q}$, where
$$
\G\cl_{p,q}=\{U\in \cl_{p,q}: U^\dagger U=e\},\quad U^\dagger:=\sum_A u_A (e_A)^{-1},
$$
such that
\begin{eqnarray}
M=U\Sigma V^\dagger,\label{SVDM}
\end{eqnarray}
where multivector $\Sigma$ belongs to the subspace $K$ of $\cl_{p,q}$, which is real span of a set of $d$ (\ref{dd}) fixed basis elements (always including the identity element $e$):
\begin{eqnarray}
\Sigma\in K:=\spn(\{e_{B_i}, i=1, \ldots, d\})=\{ \sum_{i=1}^d \lambda_i e_{B_i},\quad \lambda_i \in\BR\}.
\end{eqnarray}
\end{thm}
\begin{proof} Let us use the matrix representation $\beta'$ of $\cl_{p,q}$ from Section \ref{sectMR}. Then we use the isomorphisms (\ref{isgr}) and SVD of matrices (see Section \ref{sectSVDM}). In the cases $p-q=1, 5\mod 8$, the matrix representation is block-diagonal and we use SVD for each of two blocks. The singular values are always real and we get a real span of $d$ basis elements of $\cl_{p,q}$ with real diagonal matrix representation. 
\end{proof}

Thus, the meaning of SVD in real Clifford geometric algebra is the following: after multiplication on the left and on the right by elements of the group $\G\cl_{p,q}$ (\ref{GG}), any multivector $M\in\cl_{p,q}$ can be placed in a $d$-dimensional subspace $K$ of $\cl_{p,q}$, where $d$ is (\ref{dd}).

Note that the subspace $K$ from Theorem \ref{th1} is not unique. By changing the matrix representation, we can change the subspace $K$. But it always has dimension $d$ and contains the identity element $e$. For convenience, in the examples below we use the representation $\beta'$  from Section \ref{sectMR}.

Using (\ref{dd}) and (\ref{GGG}), we get (see the right-hand part of (\ref{SVDM}))
\begin{eqnarray}
\dim (K)+ 2\dim(\G\cl_{p,q})=
\begin{cases}
2^{n}, & \mbox{if $p-q=0, 1, 2 \mod 8$,}\\
2^{n}+2^{\frac{n-1}{2}}, & \mbox{if $p-q=3, 7 \mod 8$,}\\
2^{n}+3\cdot 2^{\frac{n-2}{2}}, & \mbox{if $p-q=4, 6 \mod 8$,}\\
2^{n}+ 3\cdot 2^{\frac{n-1}{2}}, & \mbox{if $p-q=5 \mod 8$,}
\end{cases}
\end{eqnarray}
which is greater than or equal to $\dim(\cl_{p,q})=2^n$, i.e. the number of independent coefficients of an arbitrary multivector $M\in\cl_{p,q}$. The equality holds in the cases $p-q=0, 1, 2\mod 8$ of real matrix representations.

\begin{ex} In the case $\cl_{2,0}\cong\Mat(2,\BR)$, we have
\begin{eqnarray}
&&\beta'(e)=\left(\begin{array}{cc}
      1 & 0 \\
      0 & 1
    \end{array}\right),\quad \beta'(e_1)=\left(\begin{array}{cc}
      0 & 1 \\
      1 & 0
    \end{array}\right),\label{G20}\\ &&\beta'(e_2)=\left(\begin{array}{cc}
      -1 & 0 \\
      0 & 1
    \end{array}\right),\quad  \beta'(e_{12})=\left(\begin{array}{cc}
      0 & 1 \\
      -1 & 0
    \end{array}\right).\nonumber
\end{eqnarray}
The matrices $\beta'(e)$ and $\beta'(e_2)$ are real and diagonal, we get the 2-dimensional subspace
\begin{eqnarray}
K=\spn(e, e_2).\label{KK}
\end{eqnarray}
Note that if we change the matrix representation so that the matrices for $e_1$ and $e_2$ are swapped, then the subspace $K$ will be a real span of $e$ and $e_1$. Thus, $K$ is not unique. For convenience, we choose (\ref{KK}) in this example.

Consider the multivector
$$M=5e+4e_1+3e_2\in\cl_{2,0}.$$
We can choose
$$U=V=\frac{1}{\sqrt{5}}(e-2e_{12})\in\G\cl_{2, 0}$$
with the properties
$$U^\dagger U=V^\dagger V=\frac{1}{\sqrt{5}}(e+2e_{12})\frac{1}{\sqrt{5}}(e-2e_{12})=e$$
such that
$$
M=U\Sigma V^\dagger
$$
with the element
$$
\Sigma=5(e-e_2)\in K=\spn(e, e_{2}).$$
The decomposition
$$
5e+4e_1+3e_2=\frac{1}{\sqrt{5}}(e-2e_{12})\,5(e-e_2)\,\frac{1}{\sqrt{5}}(e-2e_{12})$$
is equivalent to the matrix decomposition
$$
\left(\begin{array}{cc}
      2 & 4 \\
      4 & 8
    \end{array}\right)=\left(\begin{array}{cc}
      \frac{1}{\sqrt{5}} & -\frac{2}{\sqrt{5}} \\
      \frac{2}{\sqrt{5}} & \frac{1}{\sqrt{5}}
    \end{array}\right) \left(\begin{array}{cc}
      10 & 0 \\
      0 & 0
    \end{array}\right) \left(\begin{array}{cc}
  \frac{1}{\sqrt{5}} & \frac{2}{\sqrt{5}} \\
      -\frac{2}{\sqrt{5}} & \frac{1}{\sqrt{5}}
    \end{array}\right)
$$
using the matrix representation (\ref{G20}).

Following the comment of one of the anonymous reviewers, let us also present an explicit example for a defective multivector (the corresponding matrix representation is not diagonalizable). Consider the multivector
$$M=\frac{1}{2}(e_1+e_{12})\in\cl_{2,0}.$$
We can choose
$$U=e,\qquad V=e_1\in\G\cl_{2, 0}$$
with the properties
$$U^\dagger U=V^\dagger V=e$$
such that
$$
M=U\Sigma V^\dagger
$$
with the element
$$
\Sigma=\frac{1}{2}(e-e_2)\in K=\spn(e, e_{2}).$$
The decomposition
$$
\frac{1}{2}(e_1+e_{12})=e\,\frac{1}{2}(e-e_2)\,e_1$$
is equivalent to the matrix decomposition
$$
\left(\begin{array}{cc}
      0 & 1 \\
      0 & 0
    \end{array}\right)=\left(\begin{array}{cc}
      1 & 0 \\
      0 & 1
    \end{array}\right) \left(\begin{array}{cc}
      1 & 0 \\
      0 & 0
    \end{array}\right) \left(\begin{array}{cc}
  0 & 1 \\
     1 & 0
    \end{array}\right)
$$
using the matrix representation (\ref{G20}).

\end{ex}

\begin{ex} In the case $\cl_{1, 3}\cong \Mat(2,\BH)$, we have
\begin{eqnarray}
&&\beta'(e)=\left(\begin{array}{cc}
      1 & 0 \\
      0 & 1
    \end{array}\right), \beta'(e_1)=\left(\begin{array}{cc}
      0 & 1 \\
      1 & 0
    \end{array}\right), \beta'(e_2)=\left(\begin{array}{cc}
      i & 0 \\
      0 & -i
    \end{array}\right),\label{G13}\\ 
&&      \beta'(e_{3})=\left(\begin{array}{cc}
      j& 0 \\
      0 & -j
    \end{array}\right), \beta'(e_{4})=\left(\begin{array}{cc}
      0& -1 \\
      1 & 0
    \end{array}\right), \beta'(e_{12})=\left(\begin{array}{cc}
      0 & -i \\
      i & 0
    \end{array}\right),\nonumber\\
   &&  \beta'(e_{13})=\left(\begin{array}{cc}
      0 & -j \\
      j & 0
    \end{array}\right), \beta'(e_{14})=\left(\begin{array}{cc}
      1 & 0 \\
      0 & -1
    \end{array}\right), \beta'(e_{23})=\left(\begin{array}{cc}
      k & 0 \\
      0 & k
    \end{array}\right),\nonumber\\
   &&  \beta'(e_{24})=\left(\begin{array}{cc}
      0 & -i \\
      -i & 0
    \end{array}\right), \beta'(e_{34})=\left(\begin{array}{cc}
      0 & -j \\
      -j & 0
    \end{array}\right), \beta'(e_{123})=\left(\begin{array}{cc}
      0 & k \\
      k & 0
    \end{array}\right),\nonumber\\
    && \beta'(e_{124})=\left(\begin{array}{cc}
      -i & 0 \\
      0 & -i
    \end{array}\right),\quad \beta'(e_{134})=\left(\begin{array}{cc}
      -j & 0 \\
      0 & -j
    \end{array}\right),\nonumber\\
    &&\beta'(e_{234})=\left(\begin{array}{cc}
      0 & -k \\
      k & 0
    \end{array}\right),\quad \beta'(e_{1234})=\left(\begin{array}{cc}
      k & 0 \\
      0 & -k
    \end{array}\right).\nonumber
\end{eqnarray}
The matrices $\beta'(e)$, $\beta'(e_{14})$ are real and diagonal. We get the 2-dimensional subspace  
$$
K=\spn(e, e_{14}).
$$
Thus, an arbitrary multivector $M\in\cl_{1,3}$ with $16$ independent coefficients can be placed in the 2-dimensional subspace $\spn(e, e_{14})$ after multiplication on the left and on the right by two elements of the group $\G\cl_{1,3}$.
\end{ex}

\begin{ex} In the case $\cl_{2, 1}\cong\Mat(2,\BR)\oplus\Mat(2,\BR)$, the matrices $\beta'(e)$, $\beta'(e_{1})$, $\beta'(e_{23})$, and $\beta'(e_{123})$ are real and diagonal. We get the 4-dimensional subspace
$$
K=\spn(e, e_{1}, e_{23}, e_{123}).
$$
\end{ex}

\section{On the classical polar decomposition of real, complex, and quaternion matrices}\label{sectCPD}

Let us consider a classical polar decomposition (right and left) of arbitrary square real, complex, and quaternion matrices (for quaternion case, see \cite{quat}).

\begin{thm}
For an arbitrary $A\in\BR^{n\times n}$, there exist positive semi-definite symmetric matrices $P$ and $S\in\BR^{n\times n}$ (i.e. $P^\T=P$ and $z^\T P z \geq 0$, $\forall z\in\BR^n$; $S^\T=S$ and $z^\T S z \geq 0$, $\forall z\in\BR^n$) and matrix $W\in\OO(n)$ such that
\begin{eqnarray}
    A=WP=SW.
\end{eqnarray}
\end{thm}

Given a real symmetric matrix $P$, the following statements are equivalent:
\begin{itemize}
\item $P$ is positive semi-definite,
\item all the eigenvalues of $P$ are non-negative,
\item there exists a matrix $B$ such that $P = B^\T B$.
\end{itemize}

If we have SVD of the real matrix $A=U \Sigma V^\T$, then we can take $W=U V^\T$, $P=V\Sigma V^\T$, and $S=U\Sigma U^\T$. Note that $P=\sqrt{A^\T A}$ and $S=W P W^\T=\sqrt{A A^\T}$.

\begin{thm}
For an arbitrary $A\in\BC^{n\times n}$, there exist positive semi-definite Hermitian matrices $P$ and $S\in\BC^{n\times n}$ (i.e. $P^\HH=P$ and $z^\HH P z \geq 0$, $\forall z\in\BC^n$; $S^\HH=S$ and $z^\HH S z \geq 0$, $\forall z\in\BC^n$) and matrix $W\in\U(n)$ such that
\begin{eqnarray}
    A=WP=SW.
\end{eqnarray}
\end{thm}

Given a complex Hermitian matrix $P$, the following statements are equivalent:
\begin{itemize}
\item $P$ is positive semi-definite,
\item all the eigenvalues of $P$ are non-negative,
\item there exists a matrix $B$ such that $P = B^\HH B$.
\end{itemize}

If we have SVD of the complex matrix $A=U \Sigma V^\HH$, then we can take $W=U V^\HH$, $P=V\Sigma V^\HH$, and $S=U\Sigma U^\HH$. Note that $P=\sqrt{A^\HH A}$ and $S=W P W^\HH=\sqrt{A A^\HH}$.

\begin{thm}
For an arbitrary $A\in\BH^{n\times n}$, there exist quaternion positive semi-definite Hermitian matrices $P$ and $S\in\BH^{n\times n}$ (i.e. $P^*=P$ and $z^* P z \geq 0$, $\forall z\in\BH^n$; $S^*=S$ and $z^* S z \geq 0$, $\forall z\in\BH^n$) and matrix $W\in\Sp(n)$ such that
\begin{eqnarray}
    A=WP=SW.
\end{eqnarray}
\end{thm}

Given a quaternion Hermitian matrix $P$, the following statements are equivalent:
\begin{itemize}
\item $P$ is positive semi-definite,
\item all the eigenvalues of $P$ are non-negative,
\item there exists a matrix $B$ such that $P = B^* B$.
\end{itemize}

If we have SVD of the quaternion matrix $A=U \Sigma V^*$, then we can take $W=U V^*$, $P=V\Sigma V^*$, and $S=U\Sigma U^*$. Note that $P=\sqrt{A^* A}$ and $S=W P W^*=\sqrt{A A^*}$.

\section{Polar decomposition in GA}\label{sectPD}

In the following theorem, we present polar decomposition of an arbitrary multivector in geometric algebra $\cl_{p,q}$. Note that the statement involves only operations in $\cl_{p,q}$.

\begin{thm}[Left and right polar decomposition in GA]\label{th2} For an arbitrary multivector $M\in\cl_{p,q}$, there exist multivectors $P, S\in\cl_{p,q}$ such that
\begin{eqnarray}
&&\!\!\!\!\!\!\!\!\!\!P^\dagger=P,\qquad S^\dagger= S,\qquad U^\dagger:=\sum_A u_A (e_A)^{-1},\\
&&\!\!\!\!\!\!\!\!\!\!P=B^\dagger B,\qquad S=C^\dagger C\qquad \mbox{for some multivectors $B, C\in\cl_{p,q},$},
\end{eqnarray}
and multivector
$$W\in \G\cl_{p,q}=\{U\in \cl_{p,q}: U^\dagger U=e\}$$
such that
$$
M=WP=SW.
$$
\end{thm}
\begin{proof} The statement follows from the results of the previous sections of this paper. Namely, we use the matrix representation $\beta'$ of $\cl_{p,q}$ from Section \ref{sectMR}, the relation (\ref{sogl}) between matrix operations and Hermitian conjugation in geometric algebras, and the classical polar decomposition of matrices discussed in Section~\ref{sectCPD}.
\end{proof}
Note that
\begin{eqnarray}
P=\sqrt{M^\dagger M},\qquad S=W P W^\dagger=\sqrt{M M^\dagger}.
\end{eqnarray}
If we have the SVD of multivector $M=U\Sigma V^\dagger$ (\ref{SVDM}), then 
\begin{eqnarray}
W=U V^\dagger,\qquad P=V\Sigma V^\dagger,\qquad S=U\Sigma U^\dagger.
\end{eqnarray}

\section{The case of complexified Clifford geometric algebras}\label{sectCompl}

Let us consider the complexified Clifford geometric algebra (CGA) $\cl_{p,q}^\BC:=\BC\otimes\cl_{p,q}$ \cite{Lounesto,Bulg}. The complexified Clifford algebra is important for different applications, in particular the complexified geometric algebra $\cl^\BC_{1,3}$ is widely used in physics (see \cite{BT, unitary, Bulg}).  

An arbitrary element (multivector) $M\in\cl^\BC_{p,q}$ has the form
$$
M=\sum_A m_A e_A\in\cl^\BC_{p,q},\qquad m_A\in\BC,
$$
where $e_A$ are basis elements of the real Clifford geometric algebra $\cl_{p,q}$ (see Section \ref{sectGA}). Note that $\cl^\BC_{p,q}$ has the following basis of $2^{n+1}$ elements:
\begin{eqnarray}
e,\, ie,\, e_1,\, ie_1,\, e_2,\, i e_2,\, \ldots,\, e_{1\ldots n},\, i e_{1\ldots n}.\label{basisC}
\end{eqnarray}

In addition to the grade involution and reversion (\ref{gir}), we use the operation of complex conjugation, which takes complex conjugation only from the coordinates $m_A$ and does not change the basis elements $e_A$:
\begin{eqnarray}
\overline{M}=\sum_A \overline{m}_A e_A\in\cl^\BC_{p,q},\qquad m_A\in\BC,\qquad M\in\cl^\BC_{p,q}.
\end{eqnarray}
We have
\begin{eqnarray}
\overline{M_1 M_2}=\overline{M_1}\,\, \overline{M_2},\qquad \forall M_1, M_2\in\cl^\BC_{p,q}.
\end{eqnarray}

Let us consider an operation of Hermitian conjugation $\dagger$ in $\cl^\BC_{p,q}$  (see \cite{unitary,Bulg}):
\begin{eqnarray}
M^\dagger:=M|_{e_A \to (e_A)^{-1},\,\, m_A \to \overline{m}_A}=\sum_A \overline{m}_A (e_A)^{-1},\qquad M\in\cl^\BC_{p,q}.
\end{eqnarray}
We have the following two equivalent definitions of this operation:
\begin{eqnarray}
&&M^\dagger=\begin{cases}
e_{1\ldots p} \overline{\widetilde{M}}e_{1\ldots p}^{-1}, & \mbox{if $p$ is odd,}\\
e_{1\ldots p} \overline{\widetilde{\widehat{M}}}e_{1\ldots p}^{-1}, & \mbox{if $p$ is even,}\\
\end{cases}\\
&&M^\dagger=
\begin{cases}
e_{p+1\ldots n} \overline{\widetilde{M}}e_{p+1\ldots n}^{-1}, & \mbox{if $q$ is even,}\\
e_{p+1\ldots n} \overline{\widetilde{\widehat{M}}}e_{p+1\ldots n}^{-1}, & \mbox{if $q$ is odd.}\\
\end{cases}
\end{eqnarray}
The operation
$$(M_1, M_2):=\langle M_1^\dagger M_2 \rangle_0,\qquad M_1, M_2\in\cl^\BC_{p,q}$$
is a (positive definite) scalar product with the properties
\begin{eqnarray}
&&\!\!\!\!\!\!\!\!\!\!\!\!\!\!\!(M_1, M_2)=\overline{(M_2, M_1)},\\
&&\!\!\!\!\!\!\!\!\!\!\!\!\!\!\!(M_1+M_2, M_3)=(M_1, M_3)+(M_2, M_3),\quad (M_1, \lambda M_2)=\lambda (M_1, M_2),\\
&&\!\!\!\!\!\!\!\!\!\!\!\!\!\!\!(M, M)\geq 0,\quad \forall M\in\cl^\BC_{p,q};\qquad (M, M)=0 \Leftrightarrow M=0\label{||M||C}
\end{eqnarray}
for arbitrary multivectors $M_1, M_2, M_3\in\cl^\BC_{p,q}$ and $\lambda\in\BC$.

Using this scalar product we introduce inner product space over the field of complex numbers (unitary space) in $\cl^\BC_{p,q}$.

We have a norm 
\begin{eqnarray}
||M||:=\sqrt{(M,M)}=\sqrt{\langle M^\dagger M \rangle_0},\qquad M\in\cl^\BC_{p,q}.\label{normC}
\end{eqnarray}
with the properties
\begin{eqnarray}
&&||M||\geq 0,\quad \forall M\in\cl^\BC_{p,q};\qquad ||M||=0 \Leftrightarrow M=0,\\
&&||M_1+M_2|| \leq ||M_1||+||M_2||,\qquad \forall M_1, M_2\in\cl^\BC_{p,q},\\
&&||\lambda M||=|\lambda| ||M||,\qquad \forall M\in\cl^\BC_{p,q},\qquad \forall \lambda\in\BC.
\end{eqnarray}

Let us consider the following faithful representation (isomorphism) of the complexified geometric algebra
\begin{eqnarray}
\beta:\cl^\BC_{p,q}\quad \to\quad
\begin{cases}
    \Mat(2^{\frac{n}{2}}, \BC), &\mbox{if $n$ is even,}\\
    \Mat(2^{\frac{n-1}{2}}, \BC)\oplus\Mat(2^{\frac{n-1}{2}}, \BC), &\mbox{if $n$ is odd.}
\end{cases}\label{isomC}
\end{eqnarray}
Let us denote the size of the corresponding matrices by
\begin{eqnarray}
N:=2^{[\frac{n+1}{2}]},\label{NN}
\end{eqnarray}
where square brackets mean taking the integer part.

Let us present an explicit form of one of these representations of $\cl^\BC_{p,q}$ (we use it also for $\cl_{p,q}$ in \cite{det} and for $\cl^\BC_{p,q}$ in \cite{LMA}). We denote this fixed representation by $\beta'$. Let us consider the case $p = n$, $q = 0$. To obtain the matrix representation for another signature with $q\neq 0$, we should multiply matrices $\beta'(e_a)$, $a = p + 1, \ldots, n$ by imaginary unit $i$. For the identity element, we always use the identity matrix $\beta'(e)=I_N$ of the corresponding dimension $N$. We always take $\beta'(e_{a_1 a_2 \ldots a_k}) = \beta' (e_{a_1}) \beta' (e_{a_2}) \cdots \beta'(e_{a_k})$. In the case $n=1$, we take $\beta'(e_1)=\diag(1, -1)$. Suppose we know $\beta'_a:=\beta'(e_a)$, $a = 1, \ldots, n$ for some fixed odd $n = 2k + 1$. Then for $n = 2k + 2$, we take
the same $\beta'(e_a)$, $a = 1, \ldots , 2k + 1$, and 
$$\beta'(e_{2k+2})=\left(
    \begin{array}{cc}
      0 & I_{\frac{N}{2}}  \\
      I_{\frac{N}{2}} & 0 
    \end{array}
  \right).$$
For $n = 2k + 3$, we take
$$\beta'(e_{a})= \left(\begin{array}{cc}
      \beta'_a & 0 \\
      0 & -\beta'_a
    \end{array}
  \right),\qquad a=1, \ldots, 2k+2,$$ 
  and 
  $$\beta'(e_{2k+3})=\left(\begin{array}{cc}
      i^{k+1}\beta'_1\cdots \beta'_{2k+2} & 0 \\
      0 & -i^{k+1}\beta'_1\cdots \beta'_{2k+2} 
    \end{array}
  \right).$$
  This recursive method gives us an explicit form of the matrix representation $\beta'$ for all $n$.

Note that for this matrix representation we have
$$
(\beta'(e_a))^\HH=\eta_{aa} \beta'(e_a),\qquad a=1, \ldots, n,
$$
where $\HH$ is the Hermitian transpose of a matrix. Using the linearity, we get that Hermitian conjugation of the matrix is consistent with Hermitian conjugation of the corresponding multivector:
\begin{eqnarray}
\beta'(M^\dagger)=(\beta'(M))^\HH,\qquad M\in\cl^\BC_{p,q}.\label{soglC}
\end{eqnarray}
Note that the same is not true for an arbitrary matrix representations $\beta$ of the form (\ref{isomC}). It is true the matrix representations $\gamma=T^{-1}\beta' T$ obtained from $\beta'$ using the matrix $T$ such that $T^\HH T= I$.

Let us consider the group
\begin{eqnarray}
\G\cl^\BC_{p,q}=\{M\in \cl^\BC_{p,q}: M^\dagger M=e\},\label{GGC}
\end{eqnarray}
which we call a unitary group in $\cl^\BC_{p,q}$. Note that all the basis elements $e_A$ and $ie_A$ of $\cl^\BC_{p,q}$ belong to this group by the definition.

Using (\ref{isomC}) and (\ref{soglC}), we get the following isomorphisms to the classical matrix unitary groups:
\begin{eqnarray}
\G\cl^\BC_{p,q}\simeq\begin{cases}
    \U(2^{\frac{n}{2}}), &\mbox{if $n$ is even,}\\
    \U(2^{\frac{n-1}{2}})\times\U(2^{\frac{n-1}{2}}), &\mbox{if $n$ is odd.}
\end{cases}\label{isgrC}
\end{eqnarray}

The Lie algebra of the Lie group $\G\cl^\BC_{p,q}$ is 
$$\mathfrak{g}\cl^\BC_{p,q}=\{M\in\cl^\BC_{p,q}: M^\dagger=-M\}.$$
The basis of the Lie algebra $\mathfrak{g}\cl^\BC_{p,q}$ consists of anti-Hermitian basis elements
$$ie_1,\, ie_2,\, \ldots,\, ie_p,\, e_{p+1},\, \ldots\, e_{n},\, e_{12},\, \ldots,\, e_{p-1 p},\, ie_{p+1 p+2},\, \ldots $$
The number of such elements is equal to $2^n$; we get the dimension of the Lie group $\G\cl^\BC_{p,q}$ and the Lie algebra $\mathfrak{g}\cl^\BC_{p,q}$ (also we can calculate the dimension of the matrix Lie groups (\ref{isgrC})):
\begin{eqnarray}
\dim(\G\cl^\BC_{p,q})=\dim(\mathfrak{g}\cl^\BC_{p,q})=2^n.\label{GGGC}\end{eqnarray}

\begin{thm}[SVD in CGA]\label{th1C} For an arbitrary multivector $M\in\cl^\BC_{p,q}$, there exist multivectors $U, V\in \G\cl^\BC_{p,q}$, where
$$
\G\cl^\BC_{p,q}=\{U\in \cl^\BC_{p,q}: U^\dagger U=e\},\qquad U^\dagger:=\sum_A \overline{u}_A (e_A)^{-1},
$$
such that
\begin{eqnarray}
M=U\Sigma V^\dagger,\label{SVDMC}
\end{eqnarray}
where multivector $\Sigma$ belongs to the subspace $K\in\cl^\BC_{p,q}$, which is a real span of a set of $N=2^{[\frac{n+1}{2}]}$ fixed basis elements  (\ref{basisC}) of $\cl^\BC_{p,q}$ including the identity element $e$.
\end{thm}
\begin{proof} Let us use the matrix representation $\beta'$ of $\cl^\BC_{p,q}$ discussed above. We have the isomorphisms (\ref{isgr}) and use SVD of matrices. In the case of odd $n$, the matrix representation is block-diagonal and we use SVD for each of two blocks. The singular values are always real and we get a real span of $N$ basis elements of $\BC\otimes\cl_{p,q}$ with real diagonal matrix representation. 
\end{proof}

Thus the meaning of SVD in complexified Clifford geometric algebra is the following: after multiplication on the left and on the right by elements of the group $\G\cl^\BC_{p,q}$ (\ref{GGC}), any multivector $M\in\cl^\BC_{p,q}$ can be placed in a $N$-dimensional subspace $K\in\BC\otimes\cl_{p,q}$.

Note that the subspace $K$ from Theorem \ref{th1C} is not unique. By changing the matrix representation, we can change the subspace $K$. But it always has dimension $N$ and contains the identity element $e$. For convenience, in the examples below we use the representation $\beta'$  from this section.

Using (\ref{NN}) and (\ref{GGGC}), we get (see the right-hand part of (\ref{SVDMC}))
\begin{eqnarray}
2\dim(\G\cl^\BC_{p,q})+\dim (K)=2^{n+1}+2^{[\frac{n+1}{2}]},
\end{eqnarray}
which is always greater than $\dim(\cl^\BC_{p,q})=2^{n+1}$, i.e. the number of independent real coefficients of an arbitrary multivector $M\in\cl^\BC_{p,q}$.

\begin{ex} In the case $\cl^\BC_{2,0}\cong\Mat(2,\BC)$, we have
\begin{eqnarray}
&&\beta'(e)=\left(\begin{array}{cc}
      1 & 0 \\
      0 & 1
    \end{array}\right),\quad \beta'(e_1)=\left(\begin{array}{cc}
      1 & 0 \\
      0 & -1
    \end{array}\right),\label{G20C}\\ &&\beta'(e_2)=\left(\begin{array}{cc}
      0 & 1 \\
      1 & 0
    \end{array}\right),\quad  \beta'(e_{12})=\left(\begin{array}{cc}
      0 & 1 \\
      -1 & 0
    \end{array}\right).\nonumber
\end{eqnarray}
The matrices $\beta'(e)$ and $\beta'(e_1)$ are real and diagonal, we get the 2-dimensional subspace
\begin{eqnarray}
K=\spn(e, e_1).\label{KKK}
\end{eqnarray}
Note that if we change the matrix representation so that the matrices for $e_1$ and $e_2$ are swapped, then the subspace $K$ will be a real span of $e$ and $e_2$. Thus, $K$ is not unique. For convenience, we choose (\ref{KKK}) in this example.

Consider the multivector
\begin{eqnarray}
M=(1+i)e+(1-i)e_1+(1+i)e_2+(-1+i)e_{12}\in\cl^\BC_{2,0}.\label{M44}
\end{eqnarray}
We can choose
\begin{eqnarray}
U=\frac{1+i}{2\sqrt{2}} e+ \frac{-1+i}{2\sqrt{2}}e_1+\frac{-1+i}{2\sqrt{2}}e_2+\frac{-1-i}{2\sqrt{2}}e_{12}\in\G\cl^\BC_{2, 0},\label{U44}\\
V=\frac{1+i}{2\sqrt{2}} e+ \frac{-1+i}{2\sqrt{2}}e_1+\frac{1-i}{2\sqrt{2}}e_2+\frac{-1-i}{2\sqrt{2}}e_{12}\in\G\cl^\BC_{2, 0},\label{V44}
\end{eqnarray}
with the properties
\begin{eqnarray}
U^\dagger U=V^\dagger V=e
\end{eqnarray}
such that
\begin{eqnarray}
M=U\Sigma V^\dagger\label{PPP}
\end{eqnarray}
with the element
\begin{eqnarray}
\Sigma=2(e+e_1)\in K=\spn(e, e_{1}).\label{S44}
\end{eqnarray}
The decomposition (\ref{PPP}) is equivalent to the matrix decomposition
\begin{eqnarray}
\left(\begin{array}{cc}
      2 & 2i \\
      2 & 2i
    \end{array}\right)=\left(\begin{array}{cc}
      \frac{i}{\sqrt{2}} & \frac{-1}{\sqrt{2}} \\
      \frac{i}{\sqrt{2}} & \frac{1}{\sqrt{2}}
    \end{array}\right) \left(\begin{array}{cc}
      4 & 0 \\
      0 & 0
    \end{array}\right) \left(\begin{array}{cc}
      \frac{i}{\sqrt{2}} & \frac{-i}{\sqrt{2}} \\
      \frac{1}{\sqrt{2}} & \frac{1}{\sqrt{2}}
    \end{array}\right)^\HH
\end{eqnarray}
using the matrix representation (\ref{G20C}).
\end{ex}

\begin{ex} In the case $\cl^\BC_{3,0}\cong\Mat(2,\BC)\oplus\Mat(2,\BC)$, the matrices $\beta'(e)$, $\beta'(e_1)$,  $\beta'(ie_{23})$ and $\beta'(ie_{123})$ are diagonal and real. We get  
$$K=\spn(e, e_1, ie_{23}, ie_{123}).
$$
Thus, an arbitrary multivector $M\in\cl^\BC_{3,0}$ with $16$ independent real coefficients can be placed in the 4-dimensional real $\spn(e, e_1, ie_{23}, ie_{123})$ after multiplication on the left and on the right by two elements of the group $\G\cl^\BC_{3,0}$.
\end{ex}

\begin{thm}[Polar decomposition in CGA]\label{th2C} For an arbitrary multivector $M\in\cl^\BC_{p,q}$, there exist multivectors $P, S\in\cl^\BC_{p,q}$ such that
\begin{eqnarray}
&&\!\!\!\!\!\!\!\!\!\!P^\dagger=P,\qquad S^\dagger= S,\qquad U^\dagger:=\sum_A \overline{u}_A (e_A)^{-1},\\
&&\!\!\!\!\!\!\!\!\!\!P=B^\dagger B,\qquad S=C^\dagger C\qquad \mbox{for some multivectors $B, C\in\cl^\BC_{p,q}$},
\end{eqnarray}
and multivector
$$W\in \G\cl^\BC_{p,q}=\{U\in \cl^\BC_{p,q}: U^\dagger U=e\}$$
such that
$$
M=WP=SW.
$$
\end{thm}
\begin{proof} We use polar decomposition for matrices and results of the first part of this section about relation between complex matrices and multivectors in $\cl^\BC_{p,q}$.
\end{proof}
Note that
\begin{eqnarray}
P=\sqrt{M^\dagger M},\qquad S=W P W^\dagger=\sqrt{M M^\dagger}.
\end{eqnarray}
If we have the SVD of multivector $M=U\Sigma V^\dagger$ (\ref{SVDMC}), then 
\begin{eqnarray}
W=U V^\dagger,\qquad P=V\Sigma V^\dagger,\qquad S=U\Sigma U^\dagger.
\end{eqnarray}

\begin{ex} Let us continue the example discussed above with the multivector (\ref{M44}):
$$M=(1+i)e+(1-i)e_1+(1+i)e_2+(-1+i)e_{12}\in\cl^\BC_{2,0}.$$
Using (\ref{U44}), (\ref{V44}), and (\ref{S44}), we get the elements
\begin{eqnarray}
&&W=UV^\dagger=\frac{1}{2}(e-i e_1+i e_2-e_{12}),\\
&&P=V\Sigma V^\dagger=2(e+ie_{12}),\\
&&S=U\Sigma U^\dagger=2(e+e_2)
\end{eqnarray}
with the properties
\begin{eqnarray}
W^\dagger W=e,\qquad P^\dagger=P,\qquad S^\dagger=S.
\end{eqnarray}
We get the equalities
\begin{eqnarray}
M=WP=SW.\label{Polar44}
\end{eqnarray}
The decompositions (\ref{Polar44}) are equivalent to the matrix right and left polar decompositions
$$
\left(\begin{array}{cc}
      2 & 2i \\
      2 & 2i
    \end{array}\right)=\left(\begin{array}{cc}
      \frac{1-i}{2} & \frac{-1+i}{2} \\
      \frac{1+i}{2} & \frac{1+i}{2}
    \end{array}\right) \left(\begin{array}{cc}
      2 & 2i \\
      -2i & 2
    \end{array}\right)=
    \left(\begin{array}{cc}
      2 & 2 \\
      2 & 2
    \end{array}\right)\left(\begin{array}{cc}
      \frac{1-i}{2} & \frac{-1+i}{2} \\
      \frac{1+i}{2} & \frac{1+i}{2}
    \end{array}\right)
$$
using the matrix representation (\ref{G20C}).
\end{ex}

\section{Conclusions}\label{sectConcl}

In this paper, we naturally implement SVD and polar decomposition in real and complexified Clifford geometric algebras without using the corresponding matrix representations. Note that we use matrix representations in the proofs, namely, we use the classical SVD and polar decomposition of real, complex, and quaternion matrices. Theorems \ref{th1}, \ref{th2}, \ref{th1C}, and \ref{th2C} involve only operations in geometric algebras. The theorem on SVD in geometric algebras states that after left and right multiplication by elements of the group $\G\cl_{p,q}$ in the real case (and the group $\G\cl^\BC_{p,q}$ in the complex case), any multivector $M$ can be placed in $d$-dimensional subspace in the real case (and $N$-dimensional subspace in the complex case), where $d$ is equal to (\ref{dd}) (and $N$ is equal to (\ref{NN})). The polar decomposition is a consequence of the SVD. We expect the use of these theorems in different applications of real and complexified geometric algebras in computer science, engineering, physics, big data, machine learning, etc. This paper continues our previous research \cite{Abd,Abd2,Sylv,Sylv2,Vieta1,Vieta2} on the extension of matrix methods to geometric algebras, presented at previous ENGAGE (Empowering Novel Geometric Algebra for Graphics \& Engineering) workshops within the CGI 2020--2022 conferences.

Note that despite the statements of Theorems \ref{th1}, \ref{th2}, \ref{th1C}, and \ref{th2C} involve only operations in a geometric algebra, their proofs use matrix representation; it could be interesting to investigate, in a future work, alternative and more direct proofs involving only operations in the corresponding geometric algebra. Also note that we do not present a method (algorithm) to find the SVD in this paper. We present existing theorems. How to find elements $\Sigma$, $U$, and $V$ in (\ref{SVDM}) and (\ref{SVDMC}) using only the methods of geometric algebra and without using the corresponding matrix representations is a good and important task for further research. The problems of numerical accuracy and computation speed can also be studied. 

\section*{Acknowledgements}

The results of this paper were reported at the ENGAGE Workshop (Shanghai, China, August 2023) within the International Conference Computer Graphics International 2023 (CGI 2023). The author is grateful to the organizers and the participants of this conference for fruitful discussions.

The author is grateful to the anonymous reviewers for their careful reading of the paper and helpful comments on how to improve the presentation.

This work is supported by the Russian Science Foundation (project 23-71-10028), https://rscf.ru/en/project/23-71-10028/.

\medskip

\noindent{\bf Data availability} Data sharing not applicable to this article as no datasets were generated or analyzed during the current study.

\medskip

\noindent{\bf Declarations}\\
\noindent{\bf Conflict of interest} The authors declare that they have no conflict of interest.

\bibliographystyle{spmpsci}

\end{document}